\documentclass{llncs}

\newenvironment{proofofclaim}{\noindent {\em Proof.}}{\hfill $\diamond$ \medbreak}

\newcommand{\local}{{\cal LOCAL}}
\newcommand{\inp}{\mbox{\rm\bf x}}
\newcommand{\certif}{\mbox{\rm\bf y}}
\newcommand{\id}{\mbox{\rm Id}}
\newcommand{\LD}{\mbox{\rm LD}}
\newcommand{\NLD}{\mbox{\rm NLD}}
\newcommand{\coloring}{\mbox{\tt Coloring}}
\newcommand{\mis}{\mbox{\tt MIS}}
\newcommand{\up}{\mbox{\bf N}}

\def\cF{{\cal F}}
\def\cL{{\cal L}}

\begin{document}
%%%%%%%%%%%%%%%%%%%%%%%%%%%%%%%%%%%%%%%%%%

\title{On the Impact of Identifiers on Local Decision\thanks{This work is supported by the \emph{Jules Verne}  Franco-Icelandic bilateral scientific framework.}}

\date{}

\author{
Pierre Fraigniaud\inst{1}\thanks{
E-mail: {\tt \{pierre.fraigniaud,amos.korman\}@liafa.univ-paris-diderot.fr}.
Additional support from  ANR project DISPLEXITY, and INRIA project GANG.}
\and Magn\'us M. Halld\'orsson\inst{2}\thanks{
Supported by Iceland Research Foundation grant-of-excellence 90032021.
E-mail: {\tt mmh@ru.is}.}
\and Amos Korman$^{**}$
}

\institute{
CNRS and University Paris Diderot, France. 
\medbreak
\and
ICE-TCS, School of Computer Science, Reykjavik University, Iceland.  \\
}

\maketitle

\begin{abstract}
The issue of identifiers is crucial in distributed
computing. Informally, 
identities are used for tackling two of the fundamental difficulties that are
inherent to deterministic
distributed computing, namely: (1) {\em symmetry
  breaking}, and (2) {\em topological information gathering}. In the
context of \emph{local computation}, i.e., when nodes can gather
information only from nodes at bounded distances, some insight
regarding the role of identities has been established. For instance,
it was shown that, for large classes of \emph{construction} problems,
the role of the identities can be rather small. However, for the
identities to play no role, some other kinds of mechanisms for
breaking symmetry must be employed, 
such as 
edge-labeling or sense
of direction. When it comes to local distributed \emph{decision}
problems, the specification of the decision task does not seem to
involve symmetry breaking. Therefore, it is expected that, assuming
nodes can gather sufficient information about their neighborhood, one
could get rid of the identities, without employing extra mechanisms
for breaking symmetry. We tackle this question in the framework of the
$\local$ model.

Let $\LD$ be the class of all problems that can be \emph{decided} in a constant number of rounds in the $\local$ model. Similarly, let  $\LD^*$ be the class of all problems that can be decided at constant cost in the anonymous variant of the $\local$ model, in which nodes have no identities, but each node can get access to the (anonymous) ball of radius $t$ around it, for any $t$, at a cost of $t$. It is clear that $\LD^*\subseteq \LD$. We conjecture that $\LD^*=\LD$. In this paper, we give several evidences supporting this conjecture. In particular, we show that it holds for \emph{hereditary} problems, as well as when the
nodes know an arbitrary upper bound on the total number of nodes.  Moreover, we prove that the conjecture holds in the context of \emph{non-deterministic} local decision, where nodes are given certificates (independent of the identities, if they exist), and the decision consists in verifying these certificates. In short, we prove that $\NLD^*=\NLD$. 

\medskip


\end{abstract}

%%%%%%%%%%%%%%%%%%%%%%%%%%%%%%%%%%%%%%%%%%
\section{Introduction}
%%%%%%%%%%%%%%%%%%%%%%%%%%%%%%%%%%%%%%%%%%

\subsection{Background and Motivation}

The issue of identifiers is crucial in distributed computing~\cite{A80,NS93}. Indeed, the  correct operation of deterministic protocols often relies on the assumption that each processor $u$ comes with with a unique \emph{identity}, $\id(u)$~\cite{GHS}. Informally, in network computing, such an identity assignment is crucial for tackling two of the fundamental difficulties that are inherent to  distributed computing, namely: (1) {\em symmetry breaking}, and  (2)  {\em topological information gathering}. 

The use of identities for tackling the above two difficulties is illustrated well in the context of \emph{local} algorithms \cite{L92,L86}. Indeed, in the $\cal{LOCAL}$ model \cite{PelB00}, an algorithm that runs in $t$ communication rounds, assuming an identity assignment, can be viewed as composed of two parts: first, collecting at each node $u$, the ball $B(u,t)$ of radius $t$ around it (together with the inputs of nodes), and second, deciding the output at $u$   based solely on the information in $B(u,t)$. To achieve these two tasks, one should first obtain the ball $B(u,t)$, which may not be possible if the underlying graph is anonymous (i.e., without identities). Moreover, even if obtaining the ball is possible, e.g., if the structure of the graph allows it, the absence of unique identities  given to the nodes may prevent the algorithm from breaking symmetry. For example, in the absence of unique identities, it is impossible to design a distributed deterministic coloring algorithm, even for the symmetric connected graph composed of two nodes only.  In fact, to the best of our knowledge, all algorithms in the $\cal{LOCAL}$ model are designed assuming the presence of pairwise distinct identities or some other type of node-labeling or edge-labeling, including, e.g., sense of direction~\cite{BM09,HKP01,K09,LPR09,NS93,PS96}. 

The seminal paper of Naor and Stockmeyer \cite{NS93} provides an important insight regarding the role of identities in local computation. Informally, they show that, even though identities are necessary,  in many cases the actual values of identities is not crucial, and only their relative order matters. Specifically,  \cite{NS93} shows that for a particular class of problems, called LCL (for \emph{Locally Checkable Languages}), if there exists a local algorithm that, for any identity assignment, constructs an instance of a problem in LCL in constant number of rounds,  then there exists an \emph{order invariant}\footnote{Essentially, an order invariant algorithm uses the actual values of the identities only to impose an ordering between the nodes, that is, it behaves the same for any two identity assignments that preserve the total order between the nodes. For more details refer to \cite{NS93}.} algorithm  for that problem that  runs in the same number of rounds. LCL restricts its concern to graphs with constant maximum degree, and to problems with a constant number of inputs. The assumption on the size of the inputs of problems in LCL was shown necessary in \cite{HHRS12}, by exhibiting a natural problem that is locally checkable, has unbounded input size,  can be solved in 1 round with identities, but cannot be solved in constant time by any order invariant algorithm. The role of identities can also be gauged by comparing their impact to that of ``orientation mechanisms''.  For instance, G\"o\"os et al.~\cite{GHS12} have shown that for a large class of optimization problems, called PO-checkable problems, local algorithms do not benefit from any kind of identifiers: if a PO-checkable optimization problem can be approximated with a local algorithm, the same approximation factor can be achieved in anonymous networks if the network is provided with a port-numbering and an orientation.  

%In the discussion above, we considered general distributed
The discussion above involved distributed
\emph{construction} tasks, including, e.g., graph
coloring~\cite{BM09,K09,L92,NS93,PS96}, maximal independent
set~\cite{L92,PS96}, and maximal matching~\cite{HKP01,LPR09}. 
When it
comes to distributed \emph{decision} tasks~\cite{FKP11,FKPP12},
symmetry breaking issues do not however seem to play a
role. Informally, a decision task requires the nodes to ``collectively
decide'' whether the given instance (i.e., a graph with inputs to the
nodes) satisfies some specific properties. For instance, deciding
coloring requires, given a colored graph, to check whether this graph
is properly colored. The meaning of ``collectively decide'' is as
follows. On a legal instance, all nodes should output ``yes'', and on
an illegal one, at least one node should output ``no''. Note that it
is not really important whether this node is unique or not; hence, this
specification does not inherently require any symmetry
breaking. Therefore, assuming that each node $u$ can obtain the ball
$B(u,t)$, it makes sense that the assumption of having an identity
assignment may not be crucial for achieving correct decision.

\subsection{Model and Objectives}

We tackle the question of whether identities play a role in decision problems in the framework of the aforementioned $\local$ model~\cite{PelB00}, which is a standard distributed computing model capturing the essence of locality. Recall that, in this model, processors are nodes of a connected network $G=(V(G),E(G))$, have pairwise distinct identities, and have inputs. More formally, a {\em configuration} is a triplet $(G,\inp,\id)$ where $G$ is a connected graph, every node $v\in V(G)$ is assigned as its {\em local input} a binary string $\inp(v)\in \{0,1\}^*$, and $\id(v)$ denotes the identity of node $v$. (In some problems, the local input of every node is empty, i.e., $\inp(v)=\epsilon$ for every $v\in V(G)$, where $\epsilon$ denotes the empty binary string). Processors are woken up simultaneously, and computation proceeds over the input configuration $(G,\inp,\id)$ in fault-free synchronous \emph{rounds} during which every processor exchanges messages of unlimited size with its neighbors in the underlying network $G$, and performs arbitrary individual computations on its data. In many cases, the running time of an algorithm is measured with respect to the size $n$ of $G$: the running time of an algorithm is defined as the maximum number of rounds it takes to terminate at all nodes, over all possible $n$-node networks. Similarly to \cite{HHRS12,NS93}, we consider algorithms whose running time is independent of the size of the network, that is they run in constant time. 

Let $B(u,t)$ be the ball centered at $u$, of radius $t$, excluding the edges between two nodes at distance exactly $t$ from $u$. As mentioned before, without loss of generality, any algorithm running in time~$t=O(1)$ in the $\local$ model consists of:
\begin{enumerate}
\item  Collecting (in $t$ rounds) at every node $u$ the structure of the ball $B(u,t)$  together with all the inputs $\inp(v)$ and identities $\id(v)$ of these nodes, and,
\item Performing some individual computation at every node. 
\end{enumerate}
We define the \emph{anonymous} $\local$  model similarly to the $\local$ model, except that nodes have no identities. More precisely, an input configuration in the anonymous $\local$ model is just a pair $(G,\inp)$. An algorithm running in time~$t=O(1)$ in the anonymous $\local$ model consists of:
\begin{enumerate}
\item Getting at every node $u$ a snapshot of the structure of the ball $B(u,t)$ together with all the inputs of the nodes in this ball, and,
\item Performing some individual computation at every node. 
\end{enumerate}
Note that the \emph{anonymous} $\local$ model does not explicitly involve communications between nodes. Instead, it implicitly assumes that the underlying network supports the snapshot operation.
Clearly, this model is not stronger than the $\local$  model, and possibly even strictly weaker, since a node $u$ can no longer base its individual computation on the identities of the nodes in the ball $B(u,t)$. One can think of various other ``anonymous'' models, i.e., which do not involve node identities. In particular, there is a large literature on distributed computing in networks without node identities, where symmetry breaking is enabled thanks to locally disjoint port numbers (see, e.g., \cite{FP11}). We consider the anonymous $\local$ model to isolate the role of node identities from other symmetry breaking mechanisms.\footnote{In some sense, the anonymous $\local$ model is the strongest model among all models without node identities. Indeed, there are network problems that can be solved in the  anonymous $\local$ model  which cannot be solved in the aforementioned model that is based on locally disjoint port numbers. A simple example is to locally detect the absence of a 3-node cycle.} Our aim is to compare the power of the anonymous $\local$ model with the standard $\local$ model in order to capture the impact of identities on local distributed decision. 

Recall from~\cite{FKP11} that a {\em distributed language} is a decidable collection $\cL$ of configurations. (Since an undecidable collection of configurations remains undecidable in the distributed setting too, we consider only decidable collections of configurations). A typical example of a language is 
\[
\coloring= \{(G,\inp) \;|\; \forall v\in V(G), \forall w\in N(v), \inp(v) \neq \inp(w)\}\ ,
\]
where $N(v)$ denotes the (open) neighborhood of $v$, that is, all nodes at distance exactly~1 from $v$. Still following the terminology from~\cite{FKP11}, we say that a distributed algorithm $A$ \emph{decides} a distributed language $\cL$ if and only if for every configuration $(G,\inp)$, every node of $G$ eventually terminates and outputs ``yes'' or ``no'', satisfying the following decision rules: 
\begin{itemize}
\item 
if $(G,\inp)\in \cL$, then each node outputs ``yes'';
\item 
if $(G,\inp)\notin \cL$, then at least one node outputs ``no''.
\end{itemize}
In the (non-anonymous) $\local$ model, these two rules must be satisfied for every identity assignment. That is, all processes must output ``yes'' on a legal instance, independent of their identities. And, on an illegal instance, at least one node must output ``no'', for every identity assignment. Note that this node may potentially differ according to the identity assignment. Some languages can be decided in constant time  (e.g., $\coloring$), while others can easily be shown not to be decidable in constant time (e.g., ``is  the network planar?''). In contrast to the above examples, there are some languages whose status is unclear. To elaborate on this, consider the particular case where it is required to decide whether the network belongs to some specified family $\cF$ of graphs. If this question can be decided in a constant number of communication rounds, then this means, informally, that the family $\cF$ can somehow be characterized by relatively simple conditions. For example, a family $\cF$ of graphs that can be characterized as consisting of all graphs having no subgraph from $\cal C$, where $\cal C$ is some specified finite set of graphs, is obviously decidable in constant time. However, the question of whether a family of graphs can be characterized as above is often non-trivial. For example, characterizing cographs as precisely the graphs with no induced $P_4$, attributed to Seinsche~\cite{Seinsche74}, is not easy, and requires nontrivial usage of modular decomposition.

We are now ready to define one of our main subjects of interest, the classes $\LD$ and $\LD^*$. Specifically, $\LD$ (for \emph{local decision}) is the class of all distributed languages that can be decided by a  distributed algorithm that runs in a constant number of  rounds in the $\local$ model~\cite{FKP11}. Similarly, $\LD^*$, the anonymous version of $\LD$, is the class of all distributed languages that can be decided by a  distributed algorithm that runs in a constant number of rounds in the anonymous $\local$ model. By definition, $\LD^*\subseteq \LD$. We conjecture that $$\LD^*=\LD.$$ In this paper, we provide several evidences supporting this conjecture. In addition, we investigate the \emph{non-deterministic} version of these classes, and prove that they coincide. More specifically, a distributed {\em verification} algorithm is a distributed algorithm $A$ that gets as input, in addition to a configuration $(G,\inp)$, a global {\em certificate vector} $\certif$, i.e., every node $v$ of a graph $G$ gets as input two binary strings, an input $\inp(v)\in\{0,1\}^*$ and a certificate $\certif(v)\in\{0,1\}^*$. A verification algorithm $A$ verifies $\cL$ if and only if for every input configuration $(G,\inp)$, the following hold: 
\begin{itemize}
\item if $(G,\inp)\in \cL$, then there exists a certificate $\certif$ such that every node outputs ``yes''; 
\item if $(G,\inp)\notin \cL$, then for every certificate $\certif$, at least one node outputs ``no''.
\end{itemize}
Again, in the (non-anonymous) $\local$ model, these two rules must be satisfied for every identity assignment, but the certificates must be the same regardless of the identities.
We now recall the class $\NLD$, for \emph{non-deterministic local decision}, as defined in \cite{FKP11}: it is the class of all distributed languages that can be verified in a constant number of rounds in the $\local$ model. Similarly, we define $\NLD^*$, the anonymous version of $\NLD$, as the class of all distributed languages that can be verified in a constant number of rounds in the anonymous $\local$ model. By definition, $\NLD^*\subseteq \NLD$. 

%--------------------------------------------------------
\subsection{Our Results}
%--------------------------------------------------------

In this paper, we give several evidences supporting the conjecture $\LD^*=\LD$. In particular, we show that it holds for languages defined on paths, with a finite set of input values. More generally, we show that the conjecture holds for \emph{hereditary} languages, that is, languages closed under node deletion. Regarding arbitrary languages, and arbitrary graphs, we prove that the conjecture holds assuming that every node knows an  upper bound on the total number of nodes in the input graph. (This upper bound can be arbitrary, and may not be the same for all nodes). 

Moreover, we prove that equality between non-anonymous decision and anonymous decision holds in the context of \emph{non-deterministic} local decision, where nodes are given certificates (independent of the identities, if they exist), and the decision consists in verifying these certificates. More precisely, we prove that $\NLD^*=\NLD$. This latter result is obtained by characterizing both $\NLD$ and $\NLD^*$. 

%---------------------------------------------------------------------------
\subsection{Related Work}
%---------------------------------------------------------------------------
The question of how to locally decide (or verify)  languages has 
received quite a lot of attention recently. 
Inspired by classical computation complexity theory,
it was suggested in \cite{FKP11} that the study of decision problems may lead to new
structural insights also in the more complex distributed computing setting.
Indeed, following that paper, which focused on the $\cal{LOCAL}$ model, efforts were made to form a fundamental
computational complexity theory for distributed decision problems in
various other aspects of distributed computing
\cite{FKP11,FP12,FRT11,FRT12}.

The classes LD, NLD and BPLD defined in \cite{FKP11} are
the distributed analogues of the classes  P, NP and BPP, respectively.
The paper provides structural results,
developing a notion of local reduction and establishing completeness results.
One of the main results is the existence of a sharp threshold for randomization,
above which randomization does not help (at least for hereditary languages). 
More precisely, the BPLD classes were classified into two:
below and above the randomization threshold. 
In \cite{FKPP12}, the authors show that the hereditary assumption can be lifted if we restrict our attention
to languages on path topologies. 
These two results from \cite{FKP11,FKPP12}  are used in the current paper in a rather surprising manner. 
The authors in \cite{FKPP12} then 
``zoom''
into the spectrum of classes below the randomization threshold, and
defines a hierarchy of an infinite set of BPLD classes, each of which
is separated from the class above it in the hierarchy.

The precise knowledge of the number of nodes $n$ was shown in  \cite{FKP11} to be of large impact on non-deterministic decision. Indeed, with such a knowledge
every language can be decided non-deterministically in the model of NLD. We note, however, that the knowledge of an arbitrary upper bound on $n$ (as assumed here in one of our results) seems to be a much weaker assumption, and, in particular, will not suffice for non-deterministically deciding all languages. In the context of construction problems, it was shown in \cite{KSV11} that in many case, the knowledge of $n$ (or an upper bound on $n$) is not essential.

The original theoretical basis for non-determinism in local
computation was laid by the theory of \emph{proof-labeling schemes} (PLS)~\cite{GS11,KK07,KKM11,KKP10} 
originally defined in \cite{KKP10}. As mentioned, 
this notion resembles the notion of NLD, but differs in the role identities play. Specifically, in PLS the designer of the algorithm may base the certificates'
(called labels in the terminology of PLS) construction on the given identity assignment. In contrast, in the model of NLD, the certificates must be the same regardless of the identities of nodes.
Indeed, this difference is significant: while every language can be verified by a proof labeling scheme, not every language belongs to NLD~\cite{FKP11}. These notions also bear some similarities to the notions of \emph{local computation with advice} \cite{DP12,FGIP07,FIP10,FKL07}, 
{\em local detection}~\cite{AKY97},
{\em local checking}~\cite{APV}, or {\em silent stabilization}~\cite{silent}.
In addition, as shown later on, the notion of NLD is related also to the theory of
{\em lifts} or {\em covers}~\cite{A80,Linial01}.

Finally, the classification of decision problems in distributed
computing has been studied in several other models. For example,
\cite{DHKKNPPW} and \cite{KKP11} study specific decision problems in
the $\cal{CONGEST}$ model. In \cite{KKM11}, the authors study MST verification in the PLS sense but under the  $\cal{CONGEST}$ model of communication. 
In addition, decision problems have been studied in the asynchrony discipline
too, specifically in the framework of {\em wait-free computation}
\cite{FRT11,FRT12} and {\em mobile agents  computing} \cite{FP12}.
In the wait-free model, the main issues are not spatial constraints but timing
constraints (asynchronism and faults). The main focus of  \cite{FRT12} is
deterministic protocols aiming at studying the power of the ``decoder'',
i.e., the interpretation of the results. While this paper essentially
considers the AND-checker (since a global ``yes'' corresponds to all processes
saying ``yes''), \cite{FRT12} deals with other interpretations,
including more values (not only ``yes'' and ``no''), with the objective of
designing checkers that use the smallest number of values.

%%%%%%%%%%%%%%%%%%%%%%%%%%%%%%%%%%%%%%%%%%
\section{Deterministic Decision}
\label{sec:dd}
%%%%%%%%%%%%%%%%%%%%%%%%%%%%%%%%%%%%%%%%%%

We conjecture that $\LD=\LD^*$. A support to this conjecture is that it holds for a large class of languages, namely for all \emph{hereditary} languages, that is languages closed under node deletion. For instance, $\coloring$ and $\mis$ are hereditary, as well as all languages corresponding to hereditary graph families, such as planar graphs, interval graphs, forests, chordal graphs, cographs, perfect graphs, etc. 

\begin{theorem}\label{theo:LDhereditary}
$\LD^*=\LD$ for hereditary languages. 
\end{theorem}

To prove the theorem, it is sufficient to show that $\LD\subseteq \LD^*$ for hereditary languages. This immediately follows from the statement and proof of Theorem~3.3 in~\cite{FKP11}. Indeed, let $A$ be a non-anonymous local algorithm deciding $\cL$. This deterministic algorithm is in particular a randomized algorithm, with success probabilities $p=1$ for legal instances, and $q=1$ for illegal instance. That is, algorithm $A$ is a $(1,1)$-decider for $\cL$, according to the  definition in~\cite{FKP11}. Since $\cL$ is hereditary, and since $p^2+q>1$, the existence of $A$ implies the existence of a specific deterministic anonymous local algorithm $D$ for $\cL$. Indeed, the algorithm $D$ described in the proof of Theorem~3.3 in~\cite{FKP11}  is in fact anonymous: it simply collects the ball $B(u,t)$ of radius $t$ around each node $u$ for some constant~$t$, and $u$ then decides ``yes'' or ``no'' according to whether $B(u,t)\in\cL$ or not, regardless of the identities. 

A similar proof, based on Theorem~4.1 in~\cite{FKPP12}, enables to establish the following: 

\begin{theorem}\label{theo:LDpath}
$\LD^*=\LD$ for languages defined on the set of paths, with a finite set of input values. 
\end{theorem}

Another evidence supporting the conjecture $\LD=\LD^*$ is that it holds assuming that nodes have access to a seemingly weak oracle. Specifically, this oracle, denoted ${\up}$, simply provides each node with an arbitrarily large upper bound on the total number of nodes in the actual instance. 
(It is not assumed that all the upper bounds provided to nodes are the same). We denote by $\LD^{*\up}$ the class of languages that can be decided by an anonymous local algorithm having access to oracle $\up$, and we prove the following: 

\begin{theorem}\label{theo:LD=LD}
$\LD^*\subseteq\LD\subseteq\LD^{*\up}$.
\end{theorem}

\begin{proof}
We just need to prove that $\LD\subseteq\LD^{*\up}$. Let $\cL\in\LD$, and let $A$ be a local (non-anonymous) algorithm deciding $\cL$. Assume that the running time of $A$ is $t$. We transform $A$ into an anonymous algorithm $A'$ deciding $\cL$ in time $t$,  assuming each node $u$ in a given input $G$ has an access to the oracle $\up$, i.e., it knows an arbitrary upper bound $n_u$ on the number of nodes in $G$. 
Algorithm $A'$ works as follows. Each node $u$ collects the ball $B(u,t)$ of radius $t$ around it. Then, for every possible assignment of identities to the nodes of $B(u,t)$  taken from the range $[1,n_u]$, node $u$ simulates the behavior of the non-anonymous algorithm $A$ on the ball $B(u,t)$ with the corresponding identities. If, in one of these simulations, algorithm $A$ decides ``no'', then $A'$ decides ``no''. Otherwise, $A'$ decides ``yes''. 

We now prove the correctness of $A'$. If the input $(G,\inp)\in\cL$, then $A$ accepts it for every identity assignment to the nodes of $G$. Therefore, since, for every node $u$, every possible identity assignment to the nodes of the ball $B(u,t)$ can be extended to an  identity assignment to all the nodes of $G$, all the simulations of $A$ by $u$ return ``yes'', and hence $A'$ accepts $\cL$ as well. On the other hand, if  $(G,\inp)\notin\cL$ then $A$ rejects it for every identity assignment to the nodes of $G$. That is, for every identity assignment to the nodes of $G$, at least one node $u$ decides ``no''. (Note that, this node $u$ may be different for two different identity assignments). Let us fix one identity assignment $\id$ to the nodes of $G$, in the range $[1,n]$, and let $u$ be one node that decides ``no'' for $\id$. Let $B_{\mbox{\footnotesize $\id$}}(u,t)$ be the ball $B(u,t)$ with the identities of the nodes given by $\id$. In $A'$, since $u$ tries all possible identity assignments of the ball $B(u,t)$ in the range $[1,n_u]$ with $n\leq n_u$, in one of its simulations of $A$, node $u$ will simulate $A$ on $B_{\mbox{\footnotesize $\id$}}(u,t)$. In this simulation, node $u$ decides ``no'', and hence algorithm $A'$ rejects $\cL$ as well. 
\qed
\end{proof}

Note that the inclusion $\LD\subseteq\LD^{*\up}$ holds when one imposes no restrictions on the individual sequential running time. However, the transformation of a (non-anonymous) local algorithm into an anonymous local algorithm as described in the proof of Theorem~\ref{theo:LD=LD} is very expensive in terms of individual computation. Indeed, the number of simulations of the original local algorithm $A$ by each node $u$ can be as large as $n_u \choose n_B$ where $n_u$ is the upper bound on $n$ given by the oracle $\up$, and $n_B$ is the number of nodes in the ball $B(u,t)$. This bound can be exponential in $n$ even if the oracle provides a good approximation of $n$ (even if it gives precisely $n$). It would be nice to establish $\LD\subseteq\LD^{*\up}$ by using a transformation not involving a huge increase in the individual sequential computation time. 

%%%%%%%%%%%%%%%%%%%%%%%%%%%%%%%%%%%%%%%%%%
\section{Non-deterministic Decision}
\label{sec:ndd}
%%%%%%%%%%%%%%%%%%%%%%%%%%%%%%%%%%%%%%%%%%

In the previous section, we have seen several evidences supporting the conjecture that $\LD^*=\LD$, but whether it holds or not remains to be proved. In this section, we turn our attention to the non-deterministic variants of these two classes, and show that they coincide. More formally, we have:  

\begin{theorem}\label{theo:NLD=NLD}
$\NLD^*=\NLD$.
\end{theorem}

\begin{proof}
To prove $\NLD^*=\NLD$, it is sufficient to prove $\NLD\subseteq\NLD^* $. To establish this inclusion, we provide a sufficient condition for $\NLD^*$-membership, and prove that it is a necessary condition for  $\NLD$-membership. 

Let $I=(G,\inp)$ and $I'=(G',\inp')$ be two input instances. A \emph{homomorphism} from $I$ to $I'$ is a function $f:V(G)\to V(G')$ that preserves the edges of $G$ as well as the inputs to the nodes. Specifically, 
\[
\{u,v\}\in E(G) \Rightarrow \{f(u),f(v)\}\in E(G'),\] 
and $f$ maps every node $u\in V(G)$ to a node $f(u)\in V(G')$ satisfying $$\inp'(f(u))=\inp(u).$$ For instance, assuming the nodes have no inputs, and labeling the nodes of the $n$-node cycle $C_n$ by consecutive integers from 0 to $n-1$, modulo~$n$, then the map $f:C_8\to C_4$ defined by $f(u)=u\bmod 4$ is a homomorphism. The trivial map $g:C_8\to K_2$ defined by $g(u)=u\bmod 2$, where $K_2$ is the 2-node clique, is also a homomorphism. To establish conditions for $\NLD$- and $\NLD^*$-membership, we require the involved homomorphisms to preserve the local neighborhood of a node, and define the notion of \emph{$t$-local isomorphism}.  

Let $t$ be a positive integer. We say that $I$ is $t$-local isomorphic to $I'$ if and only if there exists an homomorphism $f$ from $I$ to $I'$ such that,  for every node $v\in V(G)$, $f$ restricted to $B_G(v,t)$ is an isomorphism from  $B_G(v,t)$ to $B_{G'}(f(v),t)$. We call such a homomorphism $f$ a $t$-local isomorphism. 

Note that a homomorphism is not necessarily a 1-local isomorphism. For instance, the aforementioned map $f:C_8\to C_4$ defined by $f(u)=u\bmod 4$ is a 1-local isomorphism, but the map $g:C_8\to K_2$ defined by $g(u)=u\bmod 2$ is not a 1-local isomorphism. 
To be a 1-local isomorphism, a homomorphism should also insure isomorphism between the balls of radius~1. Also observe that any $t$-local isomorphism $f: G \to G'$ is onto (because if a node of $G'$ has no pre-image, then neither do its neighbors have a pre-image, since  homomorphisms preserve edges, and so forth). To avoid confusion, it is thus useful to keep in mind that, informally, a $t$-local isomorphism goes from a ``larger'' graph to a ``smaller'' graph. 

\begin{definition}
\label{def:lift}
For positive integer $t$, we say that $\cL$ is \emph{$t$-closed under lift }if, for every two instances $I,I'$ such that $I$ is $t$-local isomorphic to $I'$, we have: $$I' \in\cL \Rightarrow I \in \cL.$$
\end{definition}

So, informally, Defintion~\ref{def:lift} states that, for a language $\cL$ to be $t$-closed under lift, if a ``smaller''  instance $I'$ is in $\cL$ then any ``larger'' instance $I$ that is a lift of $I'$, i.e., satisfying that $I$ is $t$-local isomorphic to $I'$, must also be in $\cL$. The following lemma gives a sufficient condition for $\NLD^*$-membership. 

\begin{lemma}\label{claim:sufficient} 
Let $\cL$ be a language. If there exists $t \geq 1$ such that $\cL$ is $t$-closed under lift, then $\cL\in\NLD^*$. 
\end{lemma} 

\begin{proofofclaim}
Let $\cL$ be a language, and assume that there exists $t \geq 1$ such that $\cL$ is $t$-closed under lift. We describe an anonymous non-deterministic local algorithm $A$ deciding $\cL$,  and performing in $t$ rounds. The certificate of each node $v$ is a triple $\certif(v)=(i,G',\inp')$ where $G'$ is an $n$-node graph with nodes labeled by distinct integers in $[1,n]=\{1,\dots,n\}$, $i\in[1,n]$, and $\inp'$ is an $n$-dimensional vector.
Informally, the certificates are interpreted by $A$ as follows. The graph $G'$ is supposed to be a ``map'' of $G$, that is, $G'$ is interpreted as an isomorphic copy of $G$. The integer $i$ is the label of the node in $G'$ corresponding to node $v$ in $G$. Finally, $\inp'$ is interpreted as the input of the nodes in $G'$. 

The algorithm $A$ performs as follows. Every node $v$ gets $B_G(v,t)$, the ball of radius $t$ around it; hence, in particular, it collects all the certificates of all the nodes at distance at most $t$ from it. Then, by comparing its own certificate with the ones of its neighbors,  it checks that the graph $G'$, and the input $\inp'$ in its certificate, are identical to the ones in the certificates of its neighbors. It also verifies consistency between the labels and the nodes in its ball of radius $t$. That is, it checks whether the labels and inputs in the certificate of the nodes in $B_G(v,t)$ are as described by its  certificate. Whenever a node fails to pass any of these tests, it outputs ``no''. Otherwise it output ``yes'' or ``no'' according to whether $(G',\inp')\in\cL$ or not, respectively. (This is doable because we are considering languages that are decidable in the usual sense of sequential computation). 

We show that $A$ performs correctly. If $(G,\inp)\in\cL$, then by labeling the nodes in $G$ by distinct integers from~1 to $|V(G)|$, and by providing the node $v$ labeled $i$ with $\certif(v)=(i,G,\inp)$, the algorithm $A$ output ``yes'' at all nodes, as desired. Consider now a instance $I=(G,\inp)\notin\cL$. Assume, for the purpose of contradiction that there exists a certificate $\certif$ leading all nodes to output ``yes''. Let $f:V(G) \to V(G')$ be defined by $f(v)=i$ where $i$ is the label of $v$ in its certificate. Since $\certif$ passes all tests of $A$, it means that (1) $\certif(v)=(i,G',\inp')$ where the instance $I'=(G',\inp')$ is the same for all nodes, (2) $f$ restricted to $B_G(v,t)$ is an isomorphism from  $B_G(v,t)$ to $B_{G'}(f(v),t)$, for every node $v$, and (3)  $(G',\inp')\in\cL$. In view of (2), $I$ is $t$-local isomorphic to $I'$. Therefore, (3) implies that $I=(G,\inp)\in\cL$, because $\cL$ is $t$-closed under lift. This is in contradiction with the actual hypothesis $(G,\inp)\notin\cL$.  Thus, for each certificate $\certif$, there must exist at least one node that outputs ``no''.  As a consequence, $A$ is a non-deterministic algorithm deciding $\cL$, and thus $\cL\in\NLD^*$.
\end{proofofclaim}

The following lemma shows that the aforementioned sufficient condition for $\NLD^*$-membership is a necessary condition for $\NLD$-membership. 

\begin{lemma}\label{claim:necessary} 
Let $\cL$ be a language. If $\cL\in\NLD$, then there exists $t \geq 1$ such that $\cL$ is $t$-closed under lift. 
\end{lemma} 

\begin{proofofclaim}
Let $\cL$ be a language in $\NLD$, and let $A$ be a non-deterministic (non-anonymous) local algorithm deciding $\cL$. Assume, for the purpose of contradiction that, for any integer $t \geq 1$, $\cL$ is \emph{not} $t$-closed under lift. That is, for any $t$, there exist two input instances $I,I'$ such that $I$ is $t$-local isomorphic to $I'$, with $I\notin\cL$ and $I'\in\cL$. Assume that $A$ runs in $t$ rounds. Without loss of generality, we can assume that $t\geq 1$.  Let $I=(G,\inp)\notin\cL$ and $I'=(G',\inp')\in\cL$ satisfying $I$ is $t$-local isomorphic to $I'$. Since $I'\in\cL$, there exists a certificate $\certif'$ such that when $A$ is running on $I'$ with certificate $\certif'$, every node output ``yes'' for every identity assignment. Since  $I$ is $t$-local isomorphic to $I'$, there exists an homomorphism $f:I\to I'$ such that, for every node $v\in G$, $f$ restricted to $B_G(v,t)$ is an isomorphism from  $B_G(v,t)$ to $B_{G'}(f(v),t)$. Let $\certif$ be the certificate for $I$ defined by $\certif(v)=\certif'(f(v))$. Consider the execution of $A$ running on $I$ with certificate $\certif$, and some arbitrary  identity assignment~$\id$. 

Since $A$ performs in $t$ rounds, the decision at each node $v$ is taken according to the inputs, certificates, and identities in the ball $B_G(v,t)$, as well as the structure of this ball. By the nature of the homomorphism $f$, and by the definition of certificate $\certif$, the structure, inputs and certificates of the ball $B_G(v,t)$, are identical to  the corresponding structure, inputs and certificates of the ball $B_{G'}(f(v),t)$. Balls may however differ in the identities of their nodes. So,  let $v_0$ be the node in $G$ deciding ``no'' for $(G,\inp)$ with certificate $\certif$. There exists such a node since $I \notin \cL$. Let $v'_0=f(v_0)$, and assign  the same identities to the nodes in $B_{G'}(v'_0,t)$  as their corresponding nodes in $B_{G}(v_0,t)$. Arbitrarily extend this identities to an identity assignment $\id'$ to  the whole graph $G'$. By doing so, the two balls are not only isomorphic, but every node in $B_{G}(v_0,t)$ has the same input, certificate and identity as its image in $B_{G'}(v'_0,t)$. Therefore, the decision taken by $A$ at $v_0\in G$ under $\id$ is the same as its decision at $v'_0\in G'$ under $\id'$. This is in contradiction to the fact that $v_0$  decides ``no'' while $v'_0$ decides ``yes''.  
\end{proofofclaim}

Lemmas~\ref{claim:sufficient} and~\ref{claim:necessary}~together establish the theorem. 
\qed
\end{proof}

The proof of Lemma~\ref{claim:sufficient} also provides an upper bound on the size of the certificates for \emph{graph languages} in $\NLD$, that is, for languages in $\NLD$ with no input. (This includes, e.g., recognition of interval graphs, and recognition of chordal graphs). Indeed, given $\cL\in\NLD$, Algorithm~$A$ in the proof of Lemma~\ref{claim:sufficient} verifies $\cL$ using a certificate at each node which is essentially an isomorphic copy of the input instance $(G,\inp)$, with nodes labeled by consecutive integers in $[1,n]$. If $\cL$ is a graph language, then there is no input $\inp$, and thus the size of the certificates depends only on the size of the graph. More precisely, we have:  

\begin{corollary}\label{cor:mlogn}
Let $\cL\in\NLD$ be a graph language. There exists an algorithm verifying $\cL$ using certificates of size  $O(n^2)$ bits at each node of every $n$-node graph in $\cL$. 
\end{corollary}

We now argue that the above bound is tight, that is, we prove the following.

\begin{proposition}
 There exists a graph language $\cL\in\NLD$ such that every algorithm verifying $\cL$ requires certificates of size $\Omega(n^2)$ bits.
\end{proposition}
\begin{proof}

Recall that  \cite{KKP10} showed that there exists a graph language for which every proof labeling scheme (PLS) requires labels of size $\Omega(n^2)$ bits (the proof of this latter result appears in a detailed version \cite{KKP10b}). 
Still in the context of PLS, \cite{GS11} showed that this lower bound holds for  two {\em natural} graph families: specifically, \cite{GS11}  showed that verifying symmetric graphs requires labels of size $\Omega(n^2)$ bits, and verifying   non-3 colorable graphs requires almost the same size of labels, specifically, $\Omega(n^2/\log n)$ bits. Note that  the certificate size required for verifying a language  in $\NLD$ is at least as large as the minimum label size required for verifying the language via a proof labeling scheme. 
Unfortunately, however, one cannot obtain our claim directly from the aforementioned results since it turns out that neither of the two graph languages (namely, symmetric graphs and non-3 colorable graphs) belongs to $\NLD$.

We therefore employ an indirect approach.
Specifically, consider a graph $G$. We say that $H$ is a {\em seed} of $G$ if there exists a 1-local isomorphism from $G$ to $H$.
Suppose $\cF$ is a family of graphs. Let \textsf{Seed-$\cF$} denote the family of graphs $G$, for which there exists a seed of  $G$ that belongs to $\cF$. 
Then, by definition, \textsf{Seed-$\cF$} is  $1$-closed under lift.  Indeed, assume that there is a 1-local isomorphism  $g$ from $G'$ to $G$, and let $H\in \cF$ be a seed of $G$ that belongs to $\cF$. Then let $f$ be the 1-local isomorphism  from $G$ to $H$. We have that $f\circ g$ is a 1-local isomorphism  from $G'$ to $H$, because, for every $u\in V(G')$, $B_{G'}(u,1)$ is isomorphic to $B_G(g(u),1)$, which in turn is isomorphic to $B_H(f(g(u)),1)$. Thus $H$ is also a seed of $G'$.  \textsf{Seed-$\cF$} is therefore in $\NLD$. Now, in the proof of corollary 2.2 in  \cite{KKP10b}, the authors construct, for every integer  $n$, a family $\cF_n$  of $n$-node graphs that requires proof labels of size $\Omega(n^2)$. Note that  for every prime integer $n'$, a graph $G$ of size $n'$ belongs to  $\cF_{n'}$ if and only if it belongs to 
\textsf{Seed-$\cF_{n'}$}.  Therefore, there exists a graph language, namely,   \textsf{Seed-$\cF_{n}$}, that requires  certificates of size $\Omega(n^2)$ bits (at least for prime $n$'s).
\qed
\end{proof}

%%%%%%%%%%%%%%%%%%%%%%%%%%%%%%%%%%%%%%%%%%
\section{Conclusion}
%%%%%%%%%%%%%%%%%%%%%%%%%%%%%%%%%%%%%%%%%%

Again, in this paper, we provide some evidences supporting the conjecture $\LD^*=\LD$. For instance, Theorem~\ref{theo:LD=LD} shows that if every node knows any upper bound on the number of nodes $n$, then all languages in $\LD$ can be decided in the anonymous $\local$ model. One interesting remark about the $\local$ model is that it is guaranteed that at least one node has an upper bound on $n$. This is for instance the case of the node with the largest identity. In the anonymous $\local$ model, however, there is no such guarantee. Finding a language whose decision would be based on the fact that one node has an upper bound on $n$ would disprove the conjecture  $\LD^*=\LD$. Nevertheless, it is not clear whether such a problem exists. 

In this paper, we also prove that $\NLD^*=\NLD$, that is, our conjecture holds for the non-deterministic setting. It is worth noticing that \cite{FKP11} proved that there exists an $\NLD$-complete problem under the local one-to-many reduction. It is not clear whether such a problem exists for $\NLD^*$. Indeed, the reduction in the completeness proof of \cite{FKP11} relies on the aforementioned guarantee that, in the $\local$ model,  at least one node has an upper bound on $n$. 

%%%%%%%%%%%%%%%%%%%%%%%%%%%%%%%%%%%%%%%%%%

%%%%%%%%%%%%%%%%%%%%%%%%%%%%%%%%%%%%%%%%%%
\end{document}